\documentclass[11pt,letter,onecolumn]{article}
\usepackage[margin=1.0in]{geometry}
\usepackage{amsfonts}
\usepackage{amsmath}
\usepackage{amssymb}
\usepackage{hyperref}
\usepackage{enumerate}
\usepackage{authblk}
\hypersetup{colorlinks=true, urlcolor=blue, linkcolor=blue}

\title{Self-reducible with easy decision version counting problems admit  additive error approximation. Connections to counting complexity, exponential time complexity, and circuit lower bounds}
\author{Eleni Bakali \\ National Technical University of Athens, CS dept. \\ mpakali@corelab.ntua.gr}

\newtheorem{theorem}{Theorem}
\newtheorem{lemma}{Lemma}
\newtheorem{definition}{Definition}
\newtheorem{corollary}{Corollary}
\newtheorem{proposition}{Proposition}

\newenvironment{proof}{\noindent\textit{Proof}. }{\hfill $\Box$}

\begin{document}
\maketitle

\begin{abstract}
%\begin{normalsize}
We consider the class of counting problems,i.e. functions in $\#$P, which are self reducible, and have easy decision version, i.e. for every input it is easy to decide if the value of the function $f(x)$ is zero. For example, $\#$independent-sets of all sizes, is such a problem, and one of the hardest of this class, since it is equivalent to $\#$SAT under multiplicative approximation preserving reductions. 

Using these two powerful properties, self reducibility and easy decision, we prove that all problems/ functions $f$ in this class can be approximated in probabilistic polynomial time within an absolute exponential error $\epsilon\cdot 2^{n'}, \forall\epsilon>0$, which for many of those problems (when $n'=n+$constant) implies additive approximation to the fraction $f(x)/2^n$. (Where $n'$ is the amount of non-determinism of some associated NPTM). 

Moreover we show that for all these problems we can have multiplicative error to the value  $f(x)$, of any desired accuracy (i.e. a RAS), in time of order $2^{2n'/3}poly(n)$, which is strictly smaller than exhaustive search. We also show that $f(x)<g(x)$ can be decided deterministically in time $g(x)poly(n), \forall g$. 

Finally we show that the Circuit Acceptance Probability Problem, which is related to derandomization and circuit lower bounds, can be solved with high probability and in polynomial time, for the family of all circuits for which the problems of counting either satisfying or unsatisfying assignments belong to TotP (which is the Karp-closure of self reducible problems with easy decision version). 
%\end{normalsize}
\end{abstract}

\section{Introduction}
In counting complexity we explore the computational complexity of functions that count the number of solutions of a decision problem. In general, the counting versions of problems are computationally more difficult than their decision versions. For example, given a DNF formula, it is easy to determine if it is satisfiable, but it seems hard to count the number of its satisfying assignments. Another example is counting independent sets of all sizes for a given graph. It is obvious that it is easy to tell if there is some independent set of any size, since a single node always is an independent set.
However it is one of the hardest problems to count, or even approximate, the number of all independent sets.

In this paper we consider the class of all self reducible problems with easy decision version. A problem here is called self reducible if the computation on a given instance can be reduced to a polynomial number of sub-instances of the same problem, and the the height of the corresponding self-reducibility tree is polynomial in the size of the input. It is proven in \cite{PZ06} that the Karp closure of self reducible with easy decision version is exactly the class TotP, which is the class of all functions $f$ for which there exist a non-deterministic polynomial time Turing Machine s.t. the number of all computation paths on input $x$ equals $f(x)+1$. A great number of problems of interest in the literature are self reducible, and many of them have easy decision version. 

Examples of such problems, from many different scientific areas are $\#$DNF-Sat, $\#$Monotone-2-Sat, $\#$Non-Cliques, $\#$NonIndependent Sets, NonNegative Permanent, Ranking,
Graph reliability,
$\#$matchings, computing the determinant of a matrix, computing the partition function of several models from statistical physics, like the Ising and the hard-core model, counting colorings of a graph with a number of colors bigger than the maximum degree, counting bases of a matroid, $\#$independent sets of all sizes \cite{SinclairNotesMC}, and many more. (Definitions of the above and references  can be found in \cite{PZ06,AB09, SinclairNotesMC}).

Since computing exactly counting problems seems hard (unless P=NP), a first question to ask about them is their approximability status.  Concerning multiplicative error, it is proven \cite{JS89} that for self reducible problems, even a polynomial multiplicative-error deterministic (or randomized) algorithm can be transformed to an FPTAS (FPRAS respectively), which means that we can approximate it to within any positive factor of approximation $\epsilon$ in time polynomial in $n,1/\epsilon$. So, for a self-reducible problem either there exist an FPTAS (respectively FPRAS), either it is non approximable to any polynomial factor unless P=NP (respectively P=RP).
 
The same holds even for problems with easy decision version. For example there is an FPRAS for $\#$DNF SAT (satisfying assignments of a DNF formula), but it is proved \cite{DGGJ03} that $\#$SAT  can be reduced to $\#$IS (independent sets of all sizes), under an approximation-preserving reduction, so since $\#$SAT is inapproximable unless NP=RP, the same holds for $\#$IS.  

So a second question to ask for such problems, especially if they are inapproximable within a multiplicative error, is whether we can achieve an additive error, for the fraction of accepting solutions over the space of solutions, (e.g. the number of independent sets over the number of all subsets of nodes, or the number of satisfying assignments over $2^n$). Even for problems that admit multiplicative approximation, such an additive approximation algorithm is not comparable to a multiplicative one, in the sense it can give either a better or a worse result, depending on the input. It is better when the number of solutions is big, and worse if the number of solutions is very small.  

We investigate this question and we give a randomized polynomial time algorithm with additive error $\epsilon, \forall\epsilon>0$, for all problems/functions in TotP.

Another question of interest is the exact (probably exponential) deterministic time of computing or approximating such problems. We show, among other things, that we can have a randomized approximation scheme (i.e. a multiplicative error as small as we want) in time $O(\epsilon^{-2} poly(n)2^{2n/3})$, which is strictly smaller than the exhaustive search solution.

Finally we have the following connection to derandomization and circuit lower bounds. There is a well studied problem, the Circuit Acceptance Probability Problem (CAPP): given a circuit $C$ with $n$ input gates, estimate the proportion of satisfying assignments, i.e. if $p=\frac{\#sat-assignments}{2^n}=\Pr_x[C(x)=1]$, find $\hat{p}=p\pm \epsilon$, for e.g. $\epsilon=1/6$. This is connected to derandomization, and to circuit lower bounds. In particular it is shown by Williams in \cite{Williams10}, that if CAPP can be solved, even non-deterministically, and even in time of order $2^{\delta n}poly(n)$ for some $\delta<1$, then NEXP$\nsubseteq$P/poly. 

We show that our algorithm can be used to solve the Circuit Acceptance Probability Problem (CAPP) in polynomial time and with high probability, for the family of all circuits for which the problems of either (a) counting the number of satisfying assignments, or (b) counting the number of unsatisfying assignments, belong to TotP. For example, CNF formulas belong to this class, as well as other kinds of circuits that we mention. We believe that this fact together with some sharpening and combinations of the proofs in some references that we will mention (in the related work section), will yield interesting, non trivial, circuit lower bounds. We have left the latter for further research.
\subsection{Our Contribution}
Until now, problems in this class are treated individually, and algorithms are designed based on the specific characteristics of each problem. We instead explore what can be done if we exploit their two basic, common structural characteristics, which are self-reducibility and easy decision version. 

Based on these properties, we present a randomized algorithm which achieves for every $\epsilon$, an additive $\epsilon$ approximation to the quantity $p=f(x)/2^{n'}$, with running time polynomial in the size of the input and in $1/ \epsilon$ (precisely $O(\epsilon^{-2})$), where $n'$ is the amount of non-determinism for that problem. This is the best we could expect in the sense that any multiplicative approximation would imply NP=RP. We also show that for many interesting problems $n'=n+constant$ and so all results hold with $n'$ substituted with $n$ (i.e. the size of the input).

Our algorithms relies on non-uniform sampling, by a Markov chain that goes back and forth in the internal nodes of the computation tree of the corresponding NPTM on the given input, (see below for details on the basic ideas).

We also show that for any function $f$ in this class we can decide deterministically if $f(x)\leq g(x)$ in time $g(x)\cdot poly(n)$, where $n$ is the size of input $x$.

We also show the following results, concerning exponential time approximability.
Our algorithm can be viewed as computing $f(x)$ with an absolute error $\epsilon\cdot 2^n, \forall\epsilon>0,$ so by setting $\epsilon$ accordingly, we get an absolute error of order $2^{(\beta+1)n/2}$ in time of order $2^{(1-\beta) n}poly(n),\forall\beta\in(0,1)$. We also show that we can have in time of order $(2^{(\beta+1)n/2}+2^{(1-\beta) n})poly(n), \forall\beta\in(0,1)$ and polynomial in $\epsilon^{-2}$ time, approximation scheme, (i.e. we can get a multiplicative error $\epsilon,\forall \epsilon>0).$ All these running times are better than the exhaustive search solutions.

Then we show how our algorithm can solve with high probability, in polynomial time, and for every $\epsilon$, the Circuit Acceptance Probability Problem, for all polynomial size circuits, for which the problems of either (a) counting the number of satisfying assignments, or (b) counting the number of unsatisfying assignments, belong to TotP, e.g. DNF formulas, CNF formulas, Monotone circuits, Tree-monotone circuits, etc.

 Concerning improvements and extensions, we have that for TotP this algorithm is the best we can achieve unless NP=RP , and we have also that this kind of approximation is impossible to be extended to $\#$P unless NEXP$\nsubseteq$P/poly. If any of these conjectures holds, our results can be viewed as a possible step towards proving that.

\subsection{The basic ideas}
A key element in our proof relies on a fact proved in \cite{PZ06} that the karp-closure of self-reducible problems with easy decision version coincides with the counting class TotP. This is the class of functions $f$ in $\#$P for which there exist a polynomial-time non-deterministic Turing Machine M, such that for every input $x$, $f(x)$ equals the total number of leaves of the computational tree of M, minus $1$. Moreover M can always be s.t. its computation tree is binary (although not full). So $f(x)$ also equals the number of internal nodes of this tree.

Our first idea is the following. Instead of trying to count the number of accepting paths/solutions for the input, try to count the number of internal nodes of the computation tree of the corresponding NPTM M. This approach doesn't take into account any special characteristics of the problem at hand, but only the structural properties we already mentioned. So it can be applied to any problem in this class.

It is worth noting that in this way we reduce a problem whose set of solutions might be of some unknown structure, difficult to determine and understand, to a problem whose 'set of solutions' (the internal nodes of the tree) has the very particular structure of some binary tree of height polynomial in $n=|x|$.

Our second idea is the following. In order to estimate the number of internal nodes of the computation tree, we could try to perform uniform sampling, e.g. with a random walk. However a random walk on a tree, in general needs time exponential in the height of the tree, (polynomial in the number of nodes), and besides that, it can be proved that uniform sampling is impossible unless NP=RP. So instead, we design a Markov Chain converging in polynomial time by construction, but whose stationary distribution, although not uniform, gives us all the information needed to estimate the number of nodes, and thus the value of the function.

\subsection{Related work-Comparisons- Open Questions}
We will give some related work, comparisons to our results, and open questions.
\subsubsection{On Counting Complexity}
Counting complexity started by Valiant in \cite{Valiant79} where he defined $\#$P and showed that computing the Permanent is  $\#$P-complete under Cook reductions. For a survey on Counting Complexity see chapter 17 in \cite{AB09}. 

As shown by Zachos et.al. in \cite{KPZ99}, Cook reductions blur structural differences between counting classes, so several classes inside $\#$P have been defined and studied. $\#$PE was defined in \cite{Pagourtzis01} by Pagourtzis as the problems in $\#$P with decision version in P, TotP was introduced in \cite{KPSZ98}(see the definition in the preliminaries section), and in \cite{PZ06} was shown that it coincides with the Karp-closure of self reducible problems in $\#$PE. Other classes related to TotP, with properties, relations, and completeness results, were studied e.g. in \cite{KPSZ98, Pagourtzis01, PZ06, KPSZ2001, BGPT, FFS91, SST95, HHKW05}.

Concerning the approximability of problems in TotP, no unified approach existed until now. Problems are studied individually, and for some of them it is shown that FPRAS exist(i.e. multiplicative approximation within any factor, in time polynomial in the size of the input and in the inverse of the error), e.g. for counting satisfying assignments of DNF formulas \cite{KLM89}, and for counting perfect matchings \cite{JS96perm}, while for other it is proved that are inapproximable unless NP=P (or NP=RP for the randomized case), e.g. $\#$IS: counting independent sets of all sizes in a graph \cite{DFJ02,Wei06}. Collections on relevant results, proofs and references can be found (among other things) in e.g. \cite{Goldreich, AB09,SinclairNotesMC, Vazirani, Snotes}.

Two significant papers are related to our work. Firstly, Sinclair et.al. in \cite{JS89} showed that for self reducible problems, FPRAS is equivalent to uniform sampling, and that a polynomial factor approximation implies FPRAS. So problems in TotP either have an FPRAS, either are inapproximable within a polynomial factor. Secondly, in \cite{DGGJ03} Goldberg et. al. defined approximation preserving reductions, and classified problems according to their (multiplicative) approximability. They also showed (among other things) that $\#$IS, which is in TotP, is under approximation preserving reductions interreducible to \#SAT, which is considered inapproximable, since its decision version is NP-complete. Also in \cite{B11} is shown by Bordewich that there exist an infinite number of approximability levels between the polynomial factor and the approximability of \#SAT, if NP$\neq$RP.

So for TotP, polynomial-factor multiplicative error in polynomial time is impossible unless NP=RP. Our results show that we can have (a) time strictly smaller than brute force for a RAS, and (b) polynomial time for additive error.

We have to note here that the result about the RAS do not extend to \#SAT through the reduction of Goldberg in \cite{DGGJ03} that we mentioned earlier, because the reduction maps a formula on $n$ variables, to a graph on $n^2$ vertices. However the additive error results extend to \#SAT for CNF formulas, through a reduction to DNF that preserves the number of variables.

As for the question whether we can have such an additive approximation for the whole $\#$P, we can't rule out this possibility, but as we will see in the discussion on the connections with circuit lower bounds, we know that this is impossible unless NEXP $\nsubseteq$ P/poly.

Another interesting open question is to find inside TotP a structural characterization of the class of problems that admit an FPRAS, i.e. to find what is the significant common property that makes them approximable.

\subsubsection{On Exponential Time Complexity}
The study of the exponential time complexity of problems in NP has started by Impagliazzo et al in \cite{IPZ01} where they showed exponential hardness results. For $\#$k-SAT there have been some algorithms in the literature.

In \cite{St85} Stockmeyer proved that polynomial time randomized algorithms exist, with access to a $\Sigma_2P$ oracle. 
In \cite{Tr16} Taxler proved randomized, constant-factor approximation algorithms, of time exponential but smaller than $2^n\cdot poly(m,n)$, where $n$ is the number of variables and $m$ the number of clauses of the input, provided that k-SAT (decision version) is solvable in $O(2^{cn} m^d)$ for some $0<c<1$ and $d\geq 1$ (i.e. provided that the ETH conjecture is false).
In \cite{T12} Thurley gives approximation scheme i.e. $\forall \epsilon>0$ can achieve  multiplicative error $\epsilon$ in time $O^*(\epsilon^{-2}c_k^n)$ for $c_k<2$.
In \cite{IMR12} Impagliazzo et al gives randomized exact counting in time $O^*(2^{(1-\frac{1}{30k})n})$.
There is also an algorithm, without theoretical guarantee, in \cite{GSS06} by Gomes et al. that implements Stockmeyer's idea with some SAT solver, with outstanding performance.

Since $\#3$-SAT is $\#$P-complete, all of TotP can be reduced to that, and so we can achieve such approximations too, using the above algorithms. 

However since counting is in general harder than decision, it is meaningful to explore the exponential time complexity of counting even for problems with decision in P, since as we saw already, they might be inapproximable in polynomial time. 

Our algorithm for TotP is better than all the above, in the sense that we can have an approximation scheme in $O(\epsilon^{-2}2^{\gamma n})$ time, for all $\gamma\in(2/3,1)$, and without call to any oracle, and without any unproven assumptions. Impagliazzo's algorithm is better than ours in the sense that it gives exact counting, and it is worse in running time. Note of course that through a reduction to 3-SAT, if the number of variables of the resulting formula is more than $n$+constant, (where $n$ would be the size of the input of the original problem) the above algorithms perform even worse w.r.t. $n$.

We note again that deterministic/ randomized polynomial time approximation scheme for TotP does not exist, unless P=NP/ NP=RP respectively . It is an open problem whether we can have something in time superpolynomial and subexponential like $n^{\log n}$.    
    
\subsection{On Circuit Lower Bounds}

Excellent surveys on circuit complexity, as was the state of the art until 2009, can be found in
\cite{BS90,AB09}. 
Afterwards, progress has been made by Williams in \cite{Wil11}, where he proved ACC circuit lower bounds for NEXP and E$^{NP}$, by finding improved algorithms for the circuit class ACC. His work was based on ideas first presented in \cite{Williams10} where he proved connections between circuit lower bounds and improved algorithms for Circuit-SAT. 

There he also proved connections between solving the Circuit Acceptance Probability Problem and circuit lower bounds. If CAPP can be solved for all circuits of polynomial size, even non-deterministically, and even in time of order $2^{\delta n}poly(n)$ for some $\delta<1$, then NEXP$\nsubseteq$P/poly. 
He also proved in \cite{Wil11, Will13} that for any circuit family ${\cal C}$ closed under composition of circuits, improved SAT algorithms imply $E^{NP} \nsubseteq {\cal C}$.  

The CAPP problem was first defined and studied in relation to derandomization and circuit lower bounds in \cite{Bar02,For01, IKW02,KC99, KRC00}. In particular in \cite{IKW02} was shown that solving the CAPP in subexponential nondeterministic time infinitely often, implies NEXP$\nsubseteq$P/poly.

We solve the CAPP in polynomial time, with high probability, for the family of all polynomial size circuits, for which the problems of (a) counting the number of satisfying assignments, or (b) counting the number of unsatisfying assignments, belong to TotP (e.g. for CNF formulas). We believe that this result together with some combinations of the proofs in the above references, can yield non-trivial lower bounds.

\section{Preliminaries}

We assume the reader is familiar with basic notions from computational complexity, like a non-deterministic Turing Machine, a boolean circuit, a CNF formula (formula in conjunctive normal form), a DNF (disjunctive normal form) formula, and the classes NP, P, RP, NEXP, EXP$^{NP}$, P/poly, $\#$P, FP. For definitions see e.g. \cite{AB09}. We also assume familiarity with some basics on Markov chains, e.g. the notion of mixing time, and the stationary distribution. See e.g. \cite{Peres}.  

We also keep the following conventions regarding the kinds of error for a value $f$. Absolute error $a$: $f\pm a$, additive error $a$: $\frac{f}{2^n}\pm a$, multiplicative error $a$: $(1\pm a)f$.

\begin{definition}
$\#$P is the class of functions $f:\{0,1\}^*\rightarrow \mathbb{N}$ for which there exists a non deterministic polynomial time Turing machine (NPTM) $M_f$ s.t. the number of accepting paths of $M_f$ on input $x$ equals $f(x)$. 

$\#$PE is the class of functions $f$ in $\#$P for which the decision version, i.e. the problem of deciding if $f(x)>0$, is in P.

TotP is the class of functions $f:\{0,1\}^*\rightarrow \mathbb{N}$ for which there exists a non deterministic polynomial time Turing machine (NPTM) $M_f$ s.t. the number of all computation paths of $M_f$ on input $x$ equals $f(x)+1$. 
\end{definition}

Note that in the definition of TotP we take into account all paths, not only accepting paths like in $\#$P.  $M_f$ doesn't need to return yes or no, but it can return anything, or just halt.

\paragraph{Important Observation}
It is proved in \cite{KPSZ2001} that if for some function there exists an NPTM of the kind described in the above definition for TotP, then for the same function there exists another NPTM with the same properties, with the additional property that the non-deterministic choices at each (non determinisitc) step are exactly $2$. We will call such an NPTM 'binary'. Observe that in this case, the computation tree has $f(x)$ internal nodes, or 'branchings', since it is binary. This fact is extremely crucial for our proofs.

TotP is a subclass of $\#$P. For a relation/problem $A$ in NP we will call 'decision version' the problem of deciding if there exist an accepting computation of some NPTM deciding problem $A$, and we will call 'counting version' the problem of counting accepting computations. For problems/functions $f$ in $\#$P, or in TotP, we will call 'decision version, the problem of deciding if $f(x)\neq 0$.

It is proved in \cite{PZ06} that TotP is exactly the Karp-closure of self reducible problems in $\#$PE, under the following notion of self reducibility.

\begin{definition}
A function $f : \Sigma^*\rightarrow \mathbb{N}$ is called poly-time self-reducible if there exist polynomials 
$r$ and $q$, and polynomial time computable functions $h : \Sigma^*\times\mathbb{N} \rightarrow \Sigma^*$, 
$g :\Sigma^*\times\mathbb{N} \rightarrow \mathbb{N}$, and $t :\Sigma^*\rightarrow\mathbb{N}$ such that for all 
$x\in\Sigma ^*$:\\
(a) $f(x) = t(x) +\sum_{i=0}^{r(|x|)} g(x,i)f(h(x,i))$, that is, $f$ can be processed recursively by reducing $x$ to $h(x,i)$ ($0 \le i\le r(|x|)$),  \\
(b) the recursion terminates after at most polynomial depth (that is, $f\big(h(...h(h(x,i_1),i_2)...,i_{q(|x|)})\big)$ can be computed in polynomial time).\\
(c) $|h(...h(h(x,i_1),i_2)...,i_{q(|x|)}|\in\mathcal{O}\big(poly(|x|)\big)$.
\end{definition}

Intuitively a function $f$ is self reducible if $f(x)$ can be efficiently reduced to computing $f(x_i)$ for some other instances $x_i$, with the condition that if we continue the same procedure recursively, the resulting recursion tree (whose nodes are the respective instances) will be of polynomial height.

Note that we will refer to this recursion tree as the 'self reducibility tree'. 

For example circuit satisfiability problems are self reducible under this notion. The number of solutions (i.e. satisfying assignments) of $C$ equals the number of solutions of $C_1$, which is $C$ with its first input gate fixed to $1$, plus the number of solutions of $C_0$, which is $C$ with the first input gate fixed to $0$.

Of course circuit satisfiability is not in $P$ (as far as NP$\neq$P), so its counting version is not in TotP.

To understand the definitions better, we will give another example of a problem in TotP, show that it is self reducible, and give the corresponding NPTM (whose number of paths on input $x$ equals $f(x)+1$). 

The problem is $\#IS$: given a graph G on n nodes, $f(G)$ is the number of independent sets of all sizes. Clearly $f$ is in TotP, as a single node is always an independent set, and the self reducibility tree can be defined as follows. $f(G)$ equals the number of independent sets containing node $1$ plus the number of those not containing $1$, so $f(G)$ is reduced to $f(G_0)+f(G_1)$, where $G_0$ is $G$ with node $1$ and its neighbourhood removed, and $G_1$ is G with node $1$ removed. We do that recursively for all sub-instances that occur. So the height of the self reducibility tree is $n$. The corresponding NPTM proceeds as follows. In each step $i$ it checks whether for the corresponding sub-instances $f(G^i_0)$ and $f(G^i_1)$ is not zero, and if both of them are, then it branches (i.e. it proceeds non deterministically), else it proceeds deterministically to that sub-instance $G^i_b$ for which $f(G^i_b)>0$, if such exists, else it halts. Finally, in order to have in total $f(G)+1$ leaves (or, equivalently, computation paths),  in the end of the whole computation, it makes one more branching in the rightmost path (the one that has no "left" choice in any level).

Note that in this case, the computation tree is exactly the same as the self reducibility tree, with one more branching at the right end. And clearly the number of non deterministic bits used by the NPTM is at most the height of the self reducibility tree plus one. This is because $f(x)$ results as a simple addition of $f$ on two sub-instances. But this is not always the case, as the definition of self reducibility is more general. On the other hand this is the case for many problems defined on graphs and circuits, like counting satisfying assignments of monotone circuits, and of DNF formulas.

\section{Approximability of TotP}
As we saw in the preliminaries section, the Karp closure of self reducible problems in $\#PE$ equals the class TotP. Since the number of all paths of a (not necessarily full) binary tree, minus one, equals the number of internal nodes of that tree, to compute a function in TotP, it suffices to compute the number of branchings of the computation tree of the corresponding NPTM. 

For a problem $f$ in TotP, on input $x$, it is easy to check whether a state of some computation of the corresponding NPTM is a branching, as follows. We can associate each internal state with the string of non-deterministic choices made to reach that state. 
Given such a string, we simulate the NPTM M with these non-deterministic choices, until M either has to make another non-deterministic choice, either it halts. In the first case we consider that state as a 'branching', in the second as a 'leaf'.

Thus, the problem of counting branches of such an NPTM in time polynomial in $|x|$, reduces to the problem of counting nodes of a subtree $S$ of the full binary tree $T$ of height $n$, containing the root of $T$ (if S is not empty), in time polynomial in $n$, where $S$ is given implicitly by some oracle or poly-time predicate that tells us for every node of $T$ if it belongs to $S$.

\begin{lemma} For any $f\in$TotP, on input $x$, computing $f(x)$ in time $poly(|x|)$ is reduced to counting nodes of a subtree $S$ of the full binary tree $T$ of height $n=poly(|x|)$, containing the root of $T$ (if S is not empty), in time polynomial in $n$, where $S$ is given implicitly by some oracle or poly-time predicate, that tells us for every node of $T$ if it belongs to $S$.
\end{lemma}

We are going to give a probabilistic algorithm that  given such a predicate for some subtree $S$, approximates the size of $S$ in time $poly(n)$. It is based on a rapidly mixing Markov chain on the nodes of $S$. We will first present the Markov chain and prove its mixing time and its stationary distribution. Then we will show how we can approximate the size of $S$, using the Markov chain for sampling from its stationary distribution.

\subsection{The Markov Chain}

We define a Markov chain, having as states the nodes of a subtree of the full binary tree.

\begin{definition}\label{the_chain} Let $S$ be a subtree of the fully binary tree $T$ of height $n$, containing the root of $T$. We define the Markov chain $P$ over the nodes of $S$, with the following  transition probabilities. \\$p(i,j)=1/2$ if $j$ is the parent of $i$, \\$p(i,j)=1/4$ if $j$ is a child of $i$, \\$p(i,j)=0$ for every other $j\neq i$, and \\$p(i,i)=1-\sum_{j\neq i}p(i,j)$.
\end{definition}

\begin{proposition}\label{the_stationary}
The stationary distribution of the above Markov chain $P$ is as follows. If $d_i$ is the depth of node $i$, i.e. its distance from the root, and $n$ the height of the tree, $\forall i, \pi(i)= \alpha 2^{n-d_i}$, where $\alpha$ is a normalizing factor, so that $\sum_{i}\pi(i)=1.$ 
\end{proposition}
\begin{proof}
It is easy to check that $\sum_{i}\pi(i)p(i,j)=\pi(j)$ 
\end{proof}

Now we will prove that $P$ is rapidly mixing, i.e. polynomial in the height of the tree $S$. The intuition is the following. The simple random walk on a tree needs time polynomial in the size of the tree, which in the worst case of a fully binary tree, it is exponential in the height of the tree. The reason is that it is difficult to go from a leaf to the root, since the probability of going downwards the levels of the tree, is double the probability of going upwards. So we designed a walk such that, on the full binary tree, the probabilities of going upwards equals the probability of going downwards. So its easy to see that the mixing time equals the time of convergence to the uniform distribution over the levels of the tree, thus polynomial to the height of the tree. (Of course what we loose is that the new walk, as we saw, does not converge to the uniform distribution over the nodes of $S$, as is the case for the simple random walk, and this is the reason we cannot get an FPRAS with this approach.)

It turns out that this Markov chain converges quickly even in the general case.
There are many ways to prove the mixing time formally, and we present one of them. We will use the following  lemma from \cite{JS89}. 

Let $\{X_t\}_{t\geq 0}$ be a Markov chain over a finite state space $\cal{X}$ with transition probabilities $p_{ij}$, $p_x^{(t)}$ be the distribution of $X_t$ when starting from state $x$, $\pi$ be the stationary distribution,  $\tau_x(\epsilon)=\min\{t:||p_x^{(t)}-\pi ||\leq\epsilon\}$ be the mixing time when starting from state $x$. An ergodic Markov chain is called time reversible if $\forall i,j\in{\cal X}, p_{ij}\pi_i=p_{ji}\pi_j $. Let $H$ be the underlying graph of the chain for which we have an edge with weight $w_{ij}=p_{ij}\pi_i=p_{ji}\pi_j$ for each $i,j\in\cal{X}$. A Markov chain is called lazy if $\forall i\in{\cal X},p_{ii}\geq \frac{1}{2} .$ In \cite{JS89} the conductance of a time reversible Markov chain is defined, as follows $\Phi(H)=\min\frac{\sum_{i\in Y,j\notin Y}w_{ij}}{\sum_{i \in Y}\pi_{i}}$, where the minimum is taken over all $Y\subseteq {\cal X}$ s.t. $0<\sum_{i\in Y}\pi_i\leq\frac{1}{2}.$

\begin{lemma}
\cite{JS89} For any lazy, time reversible Markov chain \[\tau_{x}(\epsilon)\leq  const \times \left[ \frac{1}{\Phi(H)^2}(\log\pi_x^{-1}+\log\epsilon^{-1}) \right].\]
\end{lemma}

\begin{proposition}
The mixing time of $P$, when starting from the root, is polynomial in the height of the tree $n$.
\end{proposition}  
\begin{proof} First of all, we will consider the lazy version of the Markov chain, i.e. in every step, with probability $1/2$ we do nothing, and with probability $1/2$ we follow the rules as in definition \ref{the_chain}. The mixing time of $P$ is bounded by the mixing time of its lazy version. The stationary distribution is the same. The Markov chain is time reversible, and the underlying graph is a tree with edge weights $w_{uv}=\pi_u p_{uv}=2^i\alpha\times \frac{1}{8}=2^{i-3}\alpha,$ if we suppose that $u$ is the father of $v$ and $2^i\alpha$ is the probability $\pi_u$.

Now it suffices to show that $1/\Phi(H)$ is polynomial in $n$. 

Let ${\cal X}$ be the set of the nodes of $S$, i.e. the state space of the Markov chain $P$. We will consider all possible $Y\subseteq {\cal X}$ with $0\leq \pi(Y)\leq 1/2.$ We will bound the quantity $\frac{\sum_{i\in Y,j \notin Y}w_{ij}}{\sum_{i \in Y}\pi_{i}}.$

If $Y$ is connected and does not contain the root of $S$, then it is a subtree of $S$, with root let say u, and $\pi_u=\alpha 2^k$ for some $k\in \mathbb{N}.$  We have 
\[\sum_{i\in Y, j\notin Y} w_{ij}\geq w_{u,father(u)}=2^{k-2}\alpha.\]
Now let $Y'$ be the full binary tree with root $u$ and height the same as $Y$, i.e. $k$. We have
\[\sum_{i\in Y} \pi_i \leq \sum_{i\in Y'}\pi_i= \sum_{j=0}^{k}2^{k-j}\alpha\times 2^j=2^k(k+1)\alpha\leq 2^k(n+1)\alpha \] where this comes if we sum over the levels of the tree $Y'.$
So it holds
\[\frac{\sum w_{ij}}{\sum \pi_i}\geq \frac{2^{k-2}\alpha}{2^k(n+1)\alpha}=\frac{1}{4(n+1)}\]

If $Y$ is the union of two subtrees of $S$, not containing the root of $S$, and the root of the first is an ancestor of the second's root, then the same arguments hold, where now take as $u$ the root of the first subtree.

If $Y$ is the union of $\lambda$ subtrees not containing the root of $S$, for which it holds that no one's root is an ancestor of any other's root, then we can prove a same bound as follows. Let $Y_1,...Y_{\lambda}$ be the subtrees, and $k_1,k_2,...,k_{\lambda}$ be the respective probabilities of the roots of them in the stationary distribution. Then as before
\[\sum w_{ij}\geq 2^{k_1-2}\alpha+2^{k_2-2}\alpha+...+2^{k_{\lambda}-2}\alpha\] and
\[\sum_{i\in Y} \pi_{i}=\sum_{j=1...\lambda}\sum_{i\in Y_j} \pi_i\leq 2^{k_j}(n+1)\alpha\] thus
\[\frac{\sum w_{ij}}{\sum \pi_i}\geq \frac{\alpha\sum_{j=1...\lambda} 2^{k_j-2}}{(n+1)\alpha \sum_{j=1...\lambda} 2^{k_j}}=\frac{1}{4(n+1)}.\]

If $Y$ is a subtree of $S$ containing the root of $S$, then the complement of $Y$, i.e. $S\setminus Y$ is the union of $\lambda$ subtrees of the previous form. So if we let $Y_i,k_i$ be as before, then
\[\sum w_{ij}=\alpha\sum_{j=1...\lambda} 2^{k_j-2}\] and since from hypothesis $\pi(Y)\leq 1/2$, we have
\[\sum_{i\in Y}\pi_i\leq\sum_{i\in S\setminus Y} \pi_i\leq (n+1)\alpha\sum_{j=1...\lambda} 2^{k_j}\]
thus the same bound holds again.

Finally, similar arguments imply the same bound when $Y$ is an arbitrary subset of $S$ i.e. an arbitrary union of subtrees of $S$.

In total we have
$1/\Phi(H)\leq 4(n+1).$
\end{proof}

Note that this result implies mixing time quadratic in the height of the tree, which agrees with the intuition for the full binary tree, that it should be as much as the mixing time of a simple random walk over the levels of the tree, i.e. over a chain of length $n$.

Before going on with the approximation algorithm, we will prove two properties of this Markov chain, useful for the proofs that will follow.

\begin{lemma}
Let $R$ be a binary tree of height $n$, and let $\alpha_R$ be the normalizing factor of the stationary distribution $\pi_{R}$ of the above Markov chain. It holds $\alpha_R^{-1}\leq (n+1)2^n,$ and $\pi_R(root)\geq \frac{1}{n+1}$
\label{propertiesOfP}
\end{lemma}
\begin{proof}
Let $r_i$ be the number of nodes in depth $i$.
\[1=\sum_{u\in S}\pi_R(u)=\sum_{i=0}^n\sum_{u\ in\ level\ i}\pi_R(u)=\sum_{i=0}^n r_i\alpha_R\cdot 2^{n-i}
\Rightarrow \frac{1}{\alpha_R}=\sum_{i=0}^n r_i\cdot 2^{n-i}\] which is maximized when the $r_i$'s are maximized, i.e. when the tree is full binary, in which case $r_i=2^i$ and $\alpha_R^{-1}=(n+1)2^n.$ This also implies that for the root of  $R$ it holds $\pi_R(root)=\alpha_R\cdot 2^n\geq \frac{1}{n+1}.$ 
\end{proof}

\subsection{The approximation algorithm}

Let $S,T$ be as before.
We will prove that we can approximate the number of nodes of $S$ using the previous Markov chain. The key idea is that much of the information we need is in the normalizing factor $\alpha$, and although the stationary distribution is far from being uniform over the nodes of $S$, $\alpha$ is fully determined from the probability of the root, since we proved it to be  $2^n\alpha$, where n is the height of $S$. 

Let $\pi_S$ denote the probability distribution over the nodes of $S$, as defined in proposition \ref{the_stationary}, and let $\alpha_S$ denote the associated normalizing factor.

First we will show how we can compute exactly the number of nodes of $S$ if we could somehow (e.g. with an oracle, or an algorithm) know the normalizing factor $\alpha_R$ for any subtree $R$ of $T$ containing $T$'s root.

Then we will give an approximation algorithm that relies on approximating all these factors by sampling from the stationary distribution of the Markov chain described before, and estimating the probability of the root, and from that, the corresponding $\alpha_{S_i}$.

Finally we give the total error of our algorithm.

\begin{proposition} \label{sizeOfS}
Let $S$ be a binary tree of height $n$, and $\forall i=0...n,$ let $S_i$ be the subtree of  $S$ that contains all nodes up to depth $i$, and let $\alpha_{S_i}$ be the factors defined as above. Then  
\[|S|=\frac{1}{\alpha_{S_n}}-\sum_{k=0}^{n-1}\frac{1}{\alpha_{S_k}}\]
\end{proposition}

\begin{proof}
For $i=1,...,n$ let $r_i$ be the number of nodes in depth $i$. So 
$|S|=r_0+...+r_n.$

Obviously if $S$ is not empty,
\begin{equation}
r_0=1=\frac{1}{\alpha_{S_0}}.\label{r0}
\end{equation}

We will prove that $\forall k=1...n$
\begin{equation}\label{rk} r_k=\frac{1}{\alpha_{S_k}}-2\frac{1}{\alpha_{S_{k-1}}},
\end{equation}
 so then $|S|=\frac{1}{\alpha_{S_0}}+\sum_{k=1}^n(\frac{1}{\alpha_{S_k}}-2\frac{1}{\alpha_{S_{k-1}}})=\frac{1}{\alpha_{S_n}}-\sum_{k=0}^{n-1}\frac{1}{\alpha_{S_k}}.$
 
We will prove claim (\ref{rk}) by induction.

For $k=1$ we have 
\[\sum_{u\in S_{1}}\pi_{S_{1}}(u)=1\Rightarrow
\alpha_{S_1}\cdot r_1+2 \alpha_{S_1}\cdot r_0=1\Rightarrow
r_1=\frac{1}{\alpha_{S_1}}-2 r_0=\frac{1}{\alpha_{S_1}}-2\frac{1}{\alpha_{S_0}}.\]

Suppose claim (\ref{rk}) holds for $k<i\leq n.$ We will prove it holds for $k=i.$

\[\sum_{u\in S_{i}}\pi_{S_{i}}(u)=1\Rightarrow
\sum_{k=0}^i 2^{i-k} \alpha_{S_i}\cdot r_k =1 \Rightarrow
r_i=\frac{1}{\alpha_{S_i}}-\sum_{k=0}^{i-1}2^{i-k} r_k\]
and substituting $r_k$ for $k=0,...,i-1$ by (\ref{r0}) and (\ref{rk}), we get
$r_i=\frac{1}{\alpha_{S_i}}-2\frac{1}{\alpha_{S_{i-1}}}.$ 
\end{proof}

\begin{corollary}
If we have an oracle, or a poly($n$) predicate that for any subtree $R$ gives the factor $\alpha_R$ defined as above, then we can compute exactly the number of nodes of any tree $S$ of height $n$ in poly($n$) time.
\end{corollary}

Now we can estimate $\alpha_R$ for any tree $R$ of height $n$, within $(1+\zeta)$ for any $\zeta>0$, with high probability, and in polynomial time, using the Markov chain over the nodes of $R$, given in Definition \ref{the_chain}.

\begin{proposition}\label{estimOfaR}
For any binary tree $R$ of height $n$ we can estimate $\alpha_R$, within $(1\pm\zeta)$ for any $\zeta>0$, with probability $1-\delta$ for any $\delta>0$, in time $poly(n,\zeta^{-1},\log\delta^{-1})$.
\end{proposition}

\begin{proof}
Let $R$ be a binary tree of height $n$. We can estimate $\alpha_R$ as follows.

As we saw, $\pi_R(root)=2^n \alpha_R$, and we observe that this is always $\geq\frac{1}{n+1}$ (which is the case when $R$ is full binary). So we can estimate $\pi_R(root)$ within  $(1\pm\zeta)$ for any $\zeta>0$, by sampling $m$ nodes of $R$ according to $\pi_R$ and taking, as estimate, the fraction $\hat{p}=\sum_{i=1}^m\frac{1}{m}X_{i}$, where $X_i=1$ if the $i$-th sample node was the root, else $X_i=0.$

It is known by standard variance analysis arguments that we need $m=O(\pi_R(root)\cdot \zeta^{-2})=poly(n,\zeta^{-1})$ to get \[\Pr [(1-\zeta)\pi_R(root)\leq \hat{p}\leq (1+\zeta)\pi_R(root)]\geq\frac{3}{4}\]

We can boost up this probability to $1-\delta$ for any $\delta>0$, by repeating the above sampling procedure $t=O(\log\delta^{-1})$ times, and taking as final estimate the median of the $t$ estimates computed each time.

(Proofs for the above arguments are elementary in courses on probabilistic algorithms or statistics, see e.g. in \cite{Snotes} the unbiased estimator theorem and the median trick, for detailed proofs.)

The random sampling  according to $\pi_R$ can be performed by running the Markov chain defined earlier, on the nodes of $R$. Observe that the deviation $\epsilon$ from the stationary distribution can be negligible and be absorbed into $\zeta$, with only a polynomial increase in the running time of the Markov chain.

Finally, the estimate for $\alpha_R$ is $\hat{\alpha_R}=2^{-n}\hat{p}$, and it holds
\[\Pr[(1-\zeta)\alpha_R\leq \hat{\alpha_R}\leq (1+\zeta)\alpha_R]\geq 1-\delta.\]   
\end{proof}

The final algorithm for estimating $|S|$ is as follows. We estimate $\alpha_{S_i}$ for every subtree of $S$ and we get an estimate of the size of $S$ using proposition \ref{sizeOfS}.

\begin{proposition}\label{main}
For all $\xi>0,\delta>0$ we can get an estimate $|\hat{S}|$ of $|S|$ in time $poly(n,\xi^{-1},\log\delta^{-1})$ s.t. \[\Pr[|S|-\xi 2^n\leq |\hat{S}|\leq |S|+\xi 2^n]\geq1-\delta\]
\end{proposition}
\begin{proof}
Let $\zeta=\frac{\xi}{2(n+1)}$ and $\epsilon=\frac{\zeta}{1+\zeta}$, thus 
$poly(\epsilon^{-1})=poly(\zeta^{-1})$
$=poly(n,\xi^{-1})$. 

So according to proposition \ref{estimOfaR} we have in time $poly(n,\xi^{-1},\log\delta^{-1})$ estimations $\forall i=1,...,n$
\begin{equation}\label{est-aSi} (1-\epsilon)\alpha_{S_i} \leq\hat{\alpha}_{S_i} \leq (1+\epsilon) \alpha_{S_i}.
\end{equation}

We will use proposition \ref{sizeOfS}. Let $A=\frac{1}{\alpha_{S_n}}$ and $B=\sum_{k=0}^{n-1}\frac{1}{\alpha_{S_k}},$ so $|S|=A-B$, and clearly $B\leq A.$

From (\ref{est-aSi}) we have 
$\frac{1}{1+\epsilon}A\leq \hat{A} \leq \frac{1}{1-\epsilon} \Leftrightarrow$
$(1-\zeta)A\leq \hat{A} \leq (1+\zeta) A$ and similarly
$(1-\zeta)B\leq \hat{B} \leq (1+\zeta) B.$ 

Thus 
$(1-\zeta)A-(1+\zeta)B\leq \hat{A}-\hat{B} \leq (1+\zeta)A-(1-\zeta)B \Leftrightarrow$

$A-B-\zeta (A+B) \leq \hat{A}-\hat{B} \leq A-B+\zeta (A+B),$ and since $A\geq B$, we have 

$|S|-2\zeta A\leq |\hat{S}| \leq |S|+2 \zeta A.$ And since from lemma \ref{propertiesOfP} the maximum $A$ is $2^n(n+1)$, we have 

$|S|-2\zeta(n+1)2^n \leq |\hat{S}| \leq|S|+ 2\zeta (n+1) 2^n \Leftrightarrow$

$|S|-\xi \cdot 2^n \leq |\hat{S}|\leq |S|+ \xi \cdot 2^n.$ 

\end{proof}

\begin{corollary}\label{pr-estimation}
Let $p=\frac{|S|}{2^n}.$ For all $\xi>0,\delta>0$ we can get an estimation  $\hat{p}$ in time $poly(n,\xi^{-1},\log\delta^{-1})$ s.t.
\[\Pr[p-\xi\leq\hat{p}\leq p+\xi]\geq1-\delta\] 
\end{corollary}

So since, as we already discussed, every problem in TotP reduces to the above problem of counting nodes of a tree, we proved the following theorem.

\begin{theorem}\label{main-theorem}
For any problem $f\in TotP$, with $M_f$ being its corresponding NPTM (whose total number of computation paths on input $x$ is $f(x)+1$), and with $n'$ being the number of non deterministic bits used by $M_f$ on input $x$,  $\forall \xi>0, \forall x\in\{0,1\}^n$ we can have with heigh probability, in time $O(\epsilon^{-2},poly(n))$ an estimation $\hat{f}(x)=f(x)\pm \xi\cdot 2^{n'}.$  

Also corollary \ref{pr-estimation} holds for $p=f(x)/2^{n'},$ i.e. we can have $\hat{p}=p\pm \xi, \forall \xi>0.$ 
\end{theorem}

The above theorem holds with $n$ in place of $n'$, if $n'=n+constant$, as is the case for many problems like counting non-cliques of a graph, counting independent sets of all sizes of a graph, counting non-independent sets of size k, counting satisfying assignments of DNF formulas, counting satisfying assignments of monotone circuits, e.t.c.

\section{Implications to exponential time complexity}
In what follows, let $f$ be a function in TotP, let $M$ be the corresponding NPTM for which $\forall x$ ($\#$branchings of $M(x)$)$=f(x).$ Let also $n$ be the size of the input, or some complexity parameter that we care about (e.g. the number of variables in a boolean formula or circuit), and $n'$ be the amount of non-deterministic bits, that is the height of the computation tree of $M(x)$ (where the internal nodes are the branchings i.e. the positions where $M$ makes a non-deterministic choice). Of course $n'$ is polynomial in $n$. Be careful that $n'$ here is denoted $n$ in proposition \ref{main}, as it is the height of the tree. For the results to have some meaning, we consider functions $s:\mathbb{N} \rightarrow \mathbb{N}$ that are positive, as small as we want, but at most $O(2^n)$.

We give corollaries of the main result.

\begin{corollary}\label{general-corollary}
For all $f\in TotP$, $\forall s:\mathbb{N} \rightarrow \mathbb{N}$, $\forall x\in \{0,1\}^*$, $\forall \delta \in (0,1)$, with probability $1-\delta$, in time $\frac{2^{n'}}{s(n')}poly(n,\log \delta^{-1})$, where $n'$ is as before, we can achieve an estimation $\hat{f}(x)=f(x) \pm 2^{n'/2}s(n')^{1/2}.$ For any $\beta \in (0,1)$, in time $2^{(1-\beta)n'}poly(n,\log \delta^{-1})$, we can achieve $\hat{f}(x)=f(x)\pm 2^{n'(1+\beta)/2}.$
\end{corollary}
\begin{proof}
From the proof of proposition \ref{main} and  in particular from the variance analysis arguments in proposition \ref{estimOfaR}, we can see that the actual dependence of the running time on $\xi$ is proportional to $\xi^{-2}$. So we get the first estimation by setting $\xi=\sqrt{\frac{s(n')}{2^{n'}}}$, and the second by setting $s(n')=2^{\beta n'}.$ 
\end{proof}

For the consequent corollaries, we will need the following useful fact.

\begin{theorem}\label{exact-deterministic}
For all $f\in TotP$, $x$, $s$ as before, we can decide deterministically in time $O(s(n) \cdot poly(n))$ whether $f(x)\leq s(n).$
\end{theorem}
\begin{proof}
We perform a bfs or a dfs on the computation tree of $M(x)$ (i.e. we perform exhaustive search by trying all non deterministic choices) until we encounter at most $s(n)+1$ branchings. If the tree is exhausted before that time, then obviously $f(x)\leq s(n)$, else $f(x)>s(n).$
\end{proof}

The next corollary shows that we can have a RAS (randomized approximation scheme) for every problem in TotP, in time strictly smaller than that of exhaustive search. Note that we can't have that in polynomial time, unless NP=RP.

\begin{corollary}
For all $f$, $x$, $s$, $\delta$, $n$, $n'$ as before, and for all $k\in \mathbb{R}$, with probability $1-\delta$ and in time $poly(k,n,\log\delta^{-1})(\frac{2^{n'}}{s(n')}+2^{n'/2}s(n')^{1/2}),$ we can achieve approximation $\hat{f}(x)=f(x)(1\pm \frac{1}{k}).$ 

For every $\beta\in(0,1)$, we can have a RAS in time $poly(k,n,\log \delta^{-1})(2^{(1- \beta)n'}+2^{(1+\beta)n'/2})$.

We can also have uniform sampling in the same amount of time.
\label{ras}\end{corollary}
\begin{proof}
First we check deterministically if $f(x)\leq k 2^{n'/2}s(n')^{1/2}$, in which case we get the exact value of $f(x)$. Otherwise, if $f(x) > k 2^{n'/2}s(n')^{1/2}$, we apply the initial algorithm to get $\hat{f}=f(x)\pm 2^{n'/2}s(n')^{1/2}$ which is $< f(x)\pm \frac{1}{k} f(x)=(1\pm \frac{1}{k})f(x).$ 
The running time is a result from theorem \ref{exact-deterministic} and corollary \ref{general-corollary}.

We can also have uniform sampling, since in \cite{JS89} is proved that a randomized approximation scheme can be used for uniform sampling with a polynomial overhead in the running time.
\end{proof}

Note that $n'$ in many cases, like problems on graphs, formulas, circuits etc., equals $n+constant$. Some example is the problem $\#IS$, as we discussed in the preliminaries section in detail.

Similar simple arguments hold for other problems too, so for these problems, since $n'=n+constant$, all the above corollaries hold with $n'$ substituted with $n$.

\begin{corollary}
For problems in TotP  for which $n'=n+constant$, like $\#IS$, and $\#SAT$ for DNF formulas, monotone circuits etc., all the above corollaries hold with $n'$ substituted with $n$.
\end{corollary}

We can explore whether we can extend corollary \ref{ras} for problems in $\#$P. One possible way is to find a (possibly of exponential time) approximation preserving reduction from a problem in $\#$P to a problem in TotP s.t. the amount of non-deterministic bits needed for the first doesn't increase too much with the reduction. 

Precisely, if $f$ is in $\#$P with $M_f$ being its corresponding NPTM (whose number of accepting computation paths on input $x$ is f(x)), that uses $n$ non deterministic bits, and $g$ is in TotP with $M_g$ its corresponding NPTM (whose total number of computation paths on input $x$ equals $f(x)+1$), that uses $n'$ non deterministic bits, then we have the following.

\begin{corollary}If there exists an approximation preserving reduction from a problem $f\in\#P$ to a problem $g\in TotP$, s.t. $n'< (3-\gamma)n/2)$, for some $\gamma\in(0,1)$, then for all $x\in\{0,1\}^n$,  $\delta\in (0,1)$, $k\in \mathbb{R}$, with probability $1-\delta$ and in time $t = poly(k,|x|,\log\delta^{-1})(2^{(1-\gamma) n}+2^{(1+\gamma)n/2}),$ we can achieve approximation $\hat{f}(x)=f(x)(1\pm \frac{1}{k}).$ The reduction suffices to be of time $O(t)$, and not polynomial.
\end{corollary}
\begin{proof}
Apply corollary \ref{ras} on $g$ with $\beta\geq 3-\frac{n}{n'}(3-\gamma).$
\end{proof}

Note that we took $n, n'$ to be the number of non deterministic bits, and not the sizes of the inputs, because we want to compare with the running time of the brute force solutions.

\section{Towards circuit lower bounds} 
There are two problems related to our results, that are related to derandomization and circuit lower bounds too. The first one is the Circuit Acceptance Probability Problem (CAPP) where given a boolean circuit with $n$ input gates, and size $n^c$ for some $c$, and it is asked to approximate the probability $p=\Pr_x[C(x)=1]$ within some $\epsilon > 0$, that is to find a $\hat{p}=p\pm \epsilon$. (In fact $\epsilon=1/6$ suffices for the results that follow). The second is the problem where given a circuit that has got either $0$ or $>2^{n-1}$ satisfying assignments, and it is asked if it is satisfiable. We will call it GapCSAT (gap circuit satisfiability).

Their relationship with circuit lower bounds was proved in \cite{IKW02, Williams10}. The CAPP and its relation to derandomization is studied in \cite{Bar02,For01, IKW02,KC99, KRC00}.

\begin{theorem}
\cite{Williams10} Suppose there is a superpolynomial $s(n)$ s.t. for all $c$ there is an $O(2^n \cdot poly(n^c))/s(n)$ nondeterministic algorithm for CAPP on $n variables$ and $n^c$ gates. Then $NEXP \nsubseteq P/poly$. The proof holds even if we replace CAPP with GapCSAT.
\end{theorem}

Since a randomized algorithm can be consider as a nondeterministic algorithm, our algorithm yields a solution to CAPP for subclasses of polynomial size circuits, for wich the counting version is in TotP, e.g monotone circuits, DNF formulas, and tree-monotone circuits. (The latter are circuits monotone w.r.t. a partial order whose graph is a tree, and their counting version problem is the basic TotP-complete problem under parsimonious reductions, as shown in \cite{BCPPZ}). The same holds for circuits that can be reduced to circuits  in TotP under additive-approximation preserving reductions, like CNF formulas.

\begin{corollary}
CAPP can be solved with heigh probability (and thus non deterministically) in time $poly(n,\epsilon^{-1}),\forall \epsilon >0,$ for circuits  with $n$ input gates, whose counting version is in TotP, and the height of the corresponding self-reducibility tree is $n+constant.$
\end{corollary}
\begin{proof}
This is a result of corollary \ref{pr-estimation}, where there $n$ essentially denotes the height of the self reducibility tree, as we already discussed in the previous subsection.
\end{proof}

\begin{corollary}
CAPP can be solved with heigh  probability in $poly(n,\epsilon^{-1})$, $\forall \epsilon >0$, for DNF formulas, monotone circuits, tree-monotone circuits, and CNF formulas of $poly(n)$ size and $n$ input gates.
\end{corollary}
\begin{proof}
The problem of counting satisfying assignments of DNF formulas, monotone circuits, and tree-monotone circuits, belongs to TotP. 

To see that the corresponding self reducibility tree is of height n+constant, observe that the number of sat.assignments of a DNF formula equals the sum of sat.assignments of the two DNF subformulas that result when we set the first variable to 0 and 1 respectively. 

The same holds for monotone circuits. For tree-monotone circuits the proof is more complicated, see \cite{BCPPZ}.

To see that the result holds for CNF formulas too, observe that if $\phi$ is a CNF formula, then its negation $\bar{\phi}$ can easily be transformed to a DNF $\psi$ with the same number of variables, using De Morgan's laws.

So if $p=\Pr_x[\phi(x)=1]$ and $q=\Pr_x[\psi(x)=1],$ then $p=1-q.$ 
If $\hat{q}=q\pm \epsilon$ then $\hat{p}=1-\hat{q}=p\pm \epsilon.$
\end{proof}

As for the GapCSAT problem, for circuits  whose counting version problem is in TotP, is solved in P by definition (of TotP). By our algorithm, it is also solved in randomized polynomial time, for circuits for which the problem of counting non-satisfying assignments is in TotP, like CNFs, and in particular it can be solved for any gap $\rho$, (i.e. the number of solutions is either $0$ or $>\rho2^n.$

\begin{corollary}
The GapCSAT problem, for any gap $\rho$, is in randomized polynomial time for circuits s.t. (a)counting the number of solutions is in TotP, or (b)counting the number of non-solutions is in TotP. (e.g. DNF, CNF).
\end{corollary}  
\begin{proof}
If the number of solutions are either $0$ or $>\rho 2^n,$ then the number of non solutions are either $2^n$ or $(1-\rho) 2^n$, so it suffices to apply our algorithm (theorem \ref{main-theorem}) with $\xi=\rho/2.$
\end{proof}

These results, combined with proofs of the given references, could give some lower bounds for circuits as the above. Also an additive approximation reduction even non deterministic, and even in subexponential time, from circuit sat to a problem in TotP, would give lower bounds for P/poly.

\subsection*{Acknowledgements}I want to thank Manolis Zampetakis for mentioning to me the relationship between these results and circuit lower bounds.


\begin{thebibliography}{40}

\bibitem{Valiant79}  L. G. Valiant. \textit{The complexity of computing the permanent}. Theoretical Computer Science, 8(2), 189–201, 1979.

\bibitem{KPSZ98} A. Kiayias, A. Pagourtzis, K. Sharma, and S. Zachos. \textit{The complexity of determining the order of solutions}. In Proceedings of the First Southern Symposium on Computing, Hattiesburg, Mississippi, 1998.

\bibitem{Pagourtzis01} A. Pagourtzis. \textit{On the complexity of hard counting problems with easy decision version}. In Proceedings of the 3rd Panhellenic Logic Symposium, Anogia, Crete, 2001.

\bibitem{PZ06} A. Pagourtzis, S. Zachos. \textit{The complexity of counting functions with easy decision version}. In Proceedings of the 31st International Symposium on Mathematical Foundations of Computer Science. Lecture Notes in Computer Science, vol. 4162, 741–752, 2006.

\bibitem{Pap94} C. H. Papadimitriou. \textit{Computational Complexity}. Addison-Wesley, 1994.

\bibitem{AB09} S. Arora, B. Barak. \textit{Computational Complexity: A Modern Approach}. Cambridge University Press New York, 2009.

\bibitem{DGGJ03} Martin E. Dyer, Leslie Ann Goldberg, Catherine S. Greenhill, Mark Jerrum:
\textit{The Relative Complexity of Approximate Counting Problems}. Algorithmica 38(3): 471-500 (2003)

\bibitem{JS89} Alistair Sinclair and Mark Jerrum. 1989. \textit{Approximate counting, uniform generation and rapidly mixing Markov chains.} Inf. Comput. 82, 1 (July 1989), 93-133.  

\bibitem{Snotes} Alistair Sinclair, Randomness and Computation, lecture notes, Fall 2011, \texttt{https://people.eecs.berkeley.edu/~sinclair/cs271/n10.pdf}

\bibitem{SinclairNotesMC} Alistair Sinclair, 
Markov Chain Monte Carlo: Foundations and Applications, lecture notes, Fall 2009
\texttt{https://people.eecs.berkeley.edu/~sinclair/cs294/f09.html}

\bibitem{Williams10} Ryan Williams, \textit{Improving exhaustive search implies superpolynomial lower bounds},Proceedings of the 42nd {ACM} Symposium on Theory of Computing, STOC 2010,231--240
 
\bibitem{IPZ01}Impagliazzo, R., Paturi, R., Zane, F. \textit{ Which problems have strongly exponential complexity?} Journal of Computer and System Sciences 62(4), 512–530 (2001) 6. 

\bibitem{St85} L. Stockmeyer. \textit{On approximation algorithms for $\#$P}. SIAM J. Comput. 14 (1985) 849–861.

\bibitem{Tr16}Patrick Traxler: \textit{The Relative Exponential Time Complexity of Approximate Counting Satisfying Assignments}. Algorithmica 75(2): 339-362 (2016).

\bibitem{T12} M. Thurley, \textit{An approximation algorithm for $\#$k-SAT}, in: Proc. 29th Int. Symp. on Theoretical Aspects of Computer Science (STACS), 2012, pp.78-87

\bibitem{IMR12} Impagliazzo, R., Matthews, W., Paturi, R.:\textit{A satisfiability algorithm for AC0.} In: SODA 2012, pp. 961–972 (2012)

\bibitem{GSS06}C.P. Gomes, A. Sabharwal, and B. Selman. \textit{Model counting: A new strategy for obtaining good bounds.} In 21th AAAI, pages 54–61, Boston, MA, July 2006.

\bibitem{KPZ99}A. Kiayias, A. Pagourtzis, and S. Zachos.\textit{ Cook
reductions blur structural differences between
functional complexity classes.} In: Panhellenic
Logic Symposium, 132-137, 1999.

\bibitem{B11}M. Bordewich. \textit{On the Approximation Complexity
Hierarchy}. y. In Approximation and
Online Algorithms, pages 37–46. Springer,
2011.

\bibitem{KPSZ2001}A. Kiayias, A. Pagourtzis, K. Sharma, and S.
Zachos. \textit{Acceptor-definable counting classes.}
In Proceedings of the 8th Panhellenic conference
on Informatics, 453-463, 2001.

\bibitem{BGPT}E. Bampas, A. Göbel, A. Pagourtzis, and A.
Tentes. \textit{On the connection between interval
size functions and path counting}. In Proc. of
TAMC'09, LNCS 5532: 108–117, Springer,
2009. Also in Computational Complexity journal,
doi: 10.1007/s00037-016-0137-8, pp. 1-
47, Springer, June 2016.

\bibitem{FFS91}S.A. Fenner, L.J. Fortnow and S.A. Kurz: \textit{Gap definable counting classes}, Proccedings 6th Annual Structure in Complexity Theory Conference, pp.30-42, Chicago, IL, 1991. 453,454

\bibitem{SST95}S.Saluja, K.V. Subrahmanyam and M.N. Thakur, \textit{Descriptive complexity of \#P functions}, Journal of Computer and Systems Sciences 50 (1995), 493-505.

\bibitem{Goldreich}O. Goldreich, Introduction to Complexity Theory, Lecture Notes Series of the Electronic Colloquium on Computational Complexity, 1999, Chapter 10: "\#P and approximating it." 
\texttt{http://www.eccc.uni-trier.de/eccc-local/ECCC-LectureNotes/}

\bibitem{Vazirani} V. Vazirani. \textit{Approximation algorithms.} Springer, (2001)

\bibitem{JS96perm}M. Jerrum and A. Sinclair. \textit{The Markov chain Monte-Carlo method: an
approach to approximate counting and integration.} In
Approximation
Algorithms for NP-hard Problems
(Dorit Hochbaum, ed.), PWS, pp. 482–
520, 1996

\bibitem{KLM89}R.M. Karp, M. Luby, and N. Madras. \textit{Monte-Carlo approximation al-
gorithms for enumeration problems.}
Journal of Algorithms
10: 429–448
(1989).

\bibitem{DFJ02} M.Dyer, A. Frieze and M.Jerrum. \textit{On counting independent sets in sparse graphs.} SIAM Journal on Computing 31 (2002), pp. 1527-1541.

\bibitem{Wei06}D.Weitz. \textit{Counting independent sets up to the tree threshold.} Proceedings of the 38th Annual ACM Symposium on Theory of Computing (STOC), 2006, pp.140-149.


\bibitem{HHKW05} L.A. Hemaspaandra, C.M. Homan, S. Kosub, and K.W. Wagner. \textit{The
complexity of computing the size of an interval.} Technical Report
cs.cc/0502058, ACM Computing Research Repository, February 2005.

\bibitem{Wil11} R. Williams, Non-uniform ACC circuit lower bounds, IEEE Conference on Computational
Complexity, 2011, pp. 115–125.

\bibitem{Bar02} B. Barak, A probabilistic-time hierarchy theorem for slightly non-uniform algorithms,
in Proc. RANDOM, Springer Lecture Notes in Comupt. Sci. 2483,
2002, pp. 194–208.

\bibitem{For01} L. Fortnow, Comparing notions of full derandomization, in Proc. IEEE Conference
on Computational Complexity, 2001, pp. 28–34.

\bibitem{IKW02} R. Impagliazzo, V. Kabanets, and A. Wigderson, In search of an easy witness:
Exponential time versus probabilistic polynomial time, J. Comput. System
Sci., 65 (2002), pp. 672–694.

\bibitem{KC99} V. Kabanets and J.-Y. Cai, Circuit minimization problem, in Proc. ACMSymposium
on Theory of Computing, 2000, pp. 73–79.

\bibitem{KRC00} V. Kabanets, C. Rackoff, and S. A. Cook, Efficiently approximable real-valued
functions, Electronic Colloquium on Computational Complexity, TR00-034,
2000.

\bibitem{Peres}D.Levin, Y. Peres, E. Wilmer: \textit{Markov chains and mixing times}. AMS (2009).

\bibitem{Will13}Ryan Williams:
\textit{Improving Exhaustive Search Implies Superpolynomial Lower Bounds.} SIAM J. Comput. 42(3): 1218-1244 (2013)

\bibitem{BS90} R. B. Boppana and M. Sipser. \textit{The complexity of finite functions.} In J.van Leeuwen, editor, Handbook of Theoretical Computer Science, volume 1.Elsevier and MIT Press, 1990.

\bibitem{BCPPZ} E. Bakali, A. Chalki, P. Pantavos, P. Pagourtzis, S. Zachos.  \textit{TotP completeness under parsimonious reductions}, manuscript.

\end{thebibliography}
\end{document}